\documentclass[conference]{IEEEtran}
\usepackage{amssymb}
\usepackage{verbatim}
\usepackage{graphicx,epsfig}
\usepackage{amsfonts}
\usepackage{subfigure}
\usepackage[normalem]{ulem}
\usepackage{amsmath}

\makeatletter
\def\ps@headings{%

\def\@oddfoot{}%
\def\@evenfoot{}}
\makeatother 

\def\bi{\begin{itemize}}
\def\ei{\end{itemize}}
\def\bequ{\begin{equation}}
\def\eequ{\end{equation}}

\def\benum{\begin{enumerate}}
\def\eenum{\end{enumerate}}

\begin{document}

\title{Optimal Solution for the Index Coding Problem Using Network Coding over GF(2)}
\author{Jalaluddin Qureshi, Chuan Heng Foh and Jianfei Cai  \\
School of Computer Engineering\\
Nanyang Technological University, Singapore\\
jala0001@e.ntu.edu.sg}

\maketitle

\begin{abstract}
The index coding problem is a fundamental transmission problem which
occurs in a wide range of multicast networks. Network coding over a
large finite field size has been shown to be a theoretically
efficient solution to the index coding problem. However the high
computational complexity of packet encoding and decoding over a
large finite field size, and its subsequent penalty on encoding and
decoding throughput and higher energy cost makes it unsuitable for
practical implementation in processor and energy constraint devices
like mobile phones and wireless sensors. While network coding over
GF(2) can alleviate these concerns, it comes at a tradeoff cost of
degrading throughput performance. To address this tradeoff, we
propose a throughput optimal triangular network coding scheme over
GF(2). We show that such a coding scheme can supply unlimited number
of innovative packets and the decoding involves the simple back
substitution. Such a coding scheme provides an efficient solution to
the index coding problem and its lower computation and energy cost
makes it suitable for practical implementation on devices with
limited processing and energy capacity.
\end{abstract}

\section{Introduction}

The index coding problem~\cite{Yossef06,Chaudhry08} is an instance
of packets $P = \{c_1, c_2, \ldots, c_M\}$ transmission problem over
a noiseless channel by a transmitter to multiple receivers $R =
\{R_1, R_2, ..., R_N\}$, given the subset of packets each receiver
already \emph{has} $H(R_i)\subseteq P$, and the disjoint subset of
packets each receiver \emph{wants} $W(R_i)\subseteq P$, such that
the total number of transmissions is minimized, as shown in
Fig.~\ref{fig:multicast}. It is intuitive to see that the index
coding problem can also be extended for erasure channel, where the
transmitter solves the index coding problem multiple times until all
the receivers have received the packets in their want set. An
efficient solution to the index coding problem has several
applications, for instance the index coding problem has been shown
to occur in content distribution networks~\cite{Chaudhry08},
satellite communication networks~\cite{Chaudhry08}, wireless
routing~\cite{Katti06} and wireless
multicasting~\cite{Sagduyu07,KChi10,Heide08} amongst many other.
Network coding~\cite{Katti06,Sagduyu07} has been proposed as an
efficient solution for the index coding problem. In network coding,
the transmitter generates a coded packet by linearly mapping packets
$c_m\in P$, with coefficients $g_m$, from a finite field $GF(2^q)$,
$q\in \mathbb{N}_1$, which is then transmitted to the receivers. As
we will show later, minimizing the total number of transmissions
requires that the transmitted coded packet is linearly independent
from the set of packets $H(R_i)$ for the maximum possible number of
receivers, ideally for all unsaturated receivers.

To illustrate how network coding can be beneficially applicable for
the index coding problem, consider for example a wireless network
where $R_0$ is multicasting packets $c_1$ and $c_2$ to $R_1$ and
$R_2$. However, $R_1$ receives $c_1$ but not $c_2$, whereas $R_2$
receives $c_2$ but not $c_1$. In this case, rather than
retransmitting packet $c_1$ and $c_2$ in 2 different time slots, it
is possible for the transmitter to encode the packets $c_1\oplus
c_2$ over $GF(2)$, and transmit the encoded packet in 1 time slot.
On receiving the encoded packet both the receiver can recover the
lost packet by decoding the original packet with the encoded packet.
This therefore reduces the total number of transmissions and hence
improves the network bandwidth.

\begin{figure}
\begin{center}
\includegraphics[width = 0.5\textwidth]{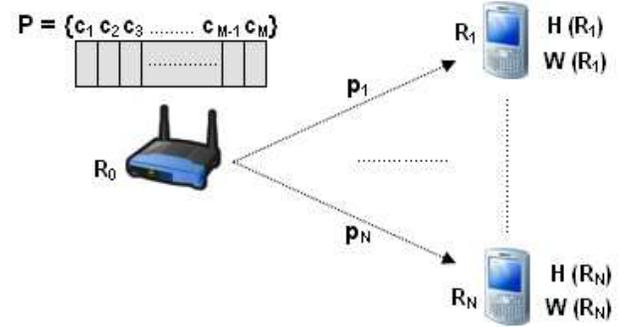}
\end{center}
\caption{Application of the index coding problem. Wireless router
multicasting a file stream to $N$ receivers.} \label{fig:multicast}
\end{figure}

However the bandwidth performance of network coding for a class of
multicast network with a packet batch size of $M$, being transmitted
to $N$ receivers, and a Bernoulli packet loss probability of $p_i$
for receiver $R_i$, has been shown to vary with the size of the
finite field $GF(2^q)$ over which packet coding is
performed~\cite{Sagduyu07,KChi10}. For random linear network coding
(RLNC), where the encoding coefficient is randomly selected from
$GF(2^q)$, the probability that the encoded packet $y_a$ is linearly
independent of all the previously received packets increases
logarithmically with respect to $q$ and is bounded by an
asymptote~\cite{Sagduyu07}. For a deterministic linear network
coding (DLNC) scheme where the encoding coefficient $g_m$, is
selected deterministically by a polynomial-time algorithm, to
guarantee that an innovative packet\footnote{An innovative packet
with respect to a specified receiver, is defined as a packet which
the receiver can't generate from the set of packets it already has.}
is transmitted at every transmission, it has been theoretically
shown that the size of the finite field from which the encoding
coefficients $g_m$, are selected is bounded as $q\geq \lceil \log_2
N\rceil$~\cite{KChi10,Kwan11}. Clearly network coding over $GF(2)$
can only guarantee an optimal solution for a network with 2
receivers. It has also been shown that the solution for the index
coding problem using network coding over $GF(2)$ is a NP-complete
problem~\cite{Rouayheb07,Kwan11}. In network coding literature,
encoding over $GF(2)$ is referred as XOR coding, whereas encoding
over $GF(2^q)$ for $q>1$ is referred to as linear network coding
(LNC), we will henceforth use these terms to distinguish between the
finite field size from which the coding coefficients are selected.

While LNC over a large field size $GF(2^q)$ has been proposed as a
theoretically viable solution to the index coding problem, such
benefits do not come without a tradeoff. The most significant
drawback of using LNC is the high computational cost of performing
packet encoding and decoding~\cite{Heide11}. Unlike XOR coding and
decoding, where the main mathematical operation is XOR addition, LNC
coding involves multiplication and addition whereas LNC decoding
involves multiplication and Gaussian elimination. LNC has packet
encoding complexity of $\mathcal {O} (MB)$~\cite{Heide08}, where $B$
is the length of the data packet. For RLNC the complexity of
generating $M$ random coefficients is given as $\mathcal {O} (MG)$,
where $G$ is the constant complexity of generating a random
number~\cite{Heide08}. For DLNC the complexity of deterministically
generating $M$ coding coefficients is given as $\mathcal {O}
(N^2M^3)$~\cite{KChi10}. The decoding computational complexity of
decoding $M$ LNC encoded packets is given as $\mathcal
{O}(M^2B+M^3)$~\cite{Heide08} for both RLNC and DLNC, and hence the
decoding complexity per coded packet is given as $\mathcal
{O}(MB+M^2)$.

The high energy cost arising because of LNC computational complexity
makes LNC unsuitable for practical implementation in battery
constrained devices like mobile phones and wireless sensors to solve
the index coding problem. Mobile phone batteries, like wireless
sensor, suffer from severe energy limitation, and energy
optimization for various smartphone applications is becoming an
important parameter for designing smartphone applications which can
sustain a longer battery life~\cite{Vallina10,Fehske11}. Using
higher capacity battery can't be regarded as a viable alternate
approach, as higher power consumption results in the overheating of
the devices. Experimental evaluation of RLNC over $GF(2^8)$ for iPod
Touch 2G has shown that packet encoding and decoding can account for
up to $33\%$ of the battery's energy consumption~\cite{Shojania09}.
Whereas it has been shown that XOR-encoding of 2 packets, each 1000
bytes long only consumes 191 nJ~\cite{Vingelmann09} of energy. Given
that transmission of a packet of the same length over IEEE 802.11
network on Nokia N95 consumes 2.31 mJ~\cite{Vingelmann09} of energy,
the overall energy cost of XOR-coding has no apparent effect on the
total energy cost of encoding and transmitting a XOR coded packet.
With energy efficiency becoming an increasingly important concern
for communication networks~\cite{Fehske11}, even if the solution for
the index coding problem using LNC is to be limited for desktop
computers and access points (AP), which can afford high energy cost,
such an approach would not be green communication feasible because
the higher energy cost of LNC corresponds to a higher carbon
footprint penalty.

LNC also suffers from low encoding and decoding throughput. The
computational penalty cost of LNC on encoding and decoding
throughput has also been practically demonstrated on a
testbed~\cite{Vingelmann10}. In this work the authors show that in
general, encoding over $GF(2)$ is approximately 8 times faster than
encoding over $GF(2^8)$ on iPhone 3G implementation. Similarly
decoding over $GF(2)$ is approximately 6 times faster than decoding
over $GF(2^8)$ on the same testbed.

From the previous discussion it is apparent that a viable solution
to the index coding problem is to use a coding scheme which can
deliver the bandwidth performance of LNC, while affording the
computation cost of XOR coding. To address this gap, in this paper
we propose a triangular pattern based packet coding scheme, where
packets are encoded over $GF(2)$, and decoding is done using the
simple back substitution scheme rather than Gaussian elimination
method. Further such coding scheme can guarantee enough `pool' on
linearly independent coded packet to deliver optimal throughput
performance. Such optimization tradeoff between energy cost and
bandwidth performance has also garnered interest
recently~\cite{Heide11}. In addition unlike deterministic coding
algorithms such as the traditional XOR-coding
schemes~\cite{Katti06,Rozner07} and DLNC polynomial-time
algorithm~\cite{KChi10} which require packet feedback information
for every transmitted packet from all the receivers, the performance
of our proposed coding scheme like RLNC is independent from the
constraint of having packet feedback information.

Our paper is organized as follow. We first highlight the current
research directions to reduce the energy cost of LNC and increase
its encoding and decoding throughput in Section~\ref{sec:related}. A
formal problem formulation and system parameters are stated in
Section~\ref{sec:problem}. We then propose our proposed coding
scheme in Section~\ref{sect:proposed}, and the evaluation of the
triangular coding throughput performance, packet overhead and
computational complexity along with comparison with current network
coding schemes in Section~\ref{sec:performance}. Numerical results
of the packet overhead of our coding scheme and its comparison with
the packet overhead of other network coding scheme are given in
Section~\ref{sect:results}. Finally we conclude with the main
contributions and results of our paper in
Section~\ref{sect:conclusion}.

\section{Related Work}\label{sec:related}
Current research direction to reduce the computation cost of LNC,
and increase the encoding and decoding throughput falls under three
broad categories. The first is the use of Gauss-Jordan elimination
method running on parallel multi-processors system as shown
in~\cite{Shojania07,Shojania09,Park10} to increase the decoding
throughput. This is based on the well-known computer science
principle that even though the Gauss-Jordan elimination method
requires more computation steps relative to Gaussian elimination,
Gauss-Jordan elimination method can nonetheless speedup the
processing time required to solve a matrix of fixed size as the
number of processors increases~\cite{Darmohray87}. This is explained
due to the better load balancing characteristics and lower
synchronization cost of the Gauss-Jordan elimination method. However
such method comes at the tradeoff cost of increasing number of
processors requirement, and higher energy cost. Even though both
Gaussian elimination and Gauss-Jordan elimination have the same
computational complexity \emph{order}, Gauss-Jordan elimination
requires more number of computation steps. A parallel Gauss-Jordan
algorithm for multi-processor system is shown to require
approximately $50\%$ more operations than Gaussian
elimination~\cite{Darmohray87}.

The second major research direction to reduce the decoding
computation cost of LNC is to use sparse coding coefficients. A
sparse matrix can loosely be described as a matrix with more number
of `0's. In~\cite{Kwan11,Sung11} it has been shown that when the
size of the coding field is bounded as $q\geq \lceil \log_2
N\rceil$, there always exists a set of coding coefficient such that
the coded packet is linearly independent from the set of packets
each of the $N$ receivers have. However, the SPARSITY
problem~\cite{Sung11} of finding such set of $M$ coding coefficients
of which $M-\omega$ coding coefficients are `0' and the coding
coefficient is linearly independent from the set of packets received
by all the receivers, is an NP-complete problem with respect to $N$.
By assuming fixed $N$, the computational complexity of solving a
$M\times M$ matrix of rank $M$ by the receiver reduces from
$\mathcal {O}(M^3)$ to $\mathcal {O} (M^2\omega)$, where $\omega\leq
M$. Unfortunately such reduction in the complexity of Gaussian
elimination come at the tradeoff cost of solving the SPARSITY
problem by the transmitter which has complexity given as $\mathcal
{O}(M^N(MN^2))$~\cite{Sung11}.

The third major approach to reduce the computation cost is to use
the trivial approach of using smaller packet batch size
(see~\cite{Li11} and references therein). However such an approach
of decreasing the packet batch size comes at the tradeoff cost of
decreasing throughput performance~\cite{Rozner07,Sagduyu07,Heide09}.

\section{Problem Statement}\label{sec:problem}
Similar to~\cite{Sagduyu07,Rouayheb07,Rozner07} we consider the
problem of a wireless transmitter $R_0$ multicasting $M$ packets to
$N$ receivers. $M$ is called the packet batch size. The \emph{iid}
Bernoulli packet loss at each receiver $R_i$ is given as $p_i$. At
the start of the transmission, $W(R_i)=P, \forall i.$ After each
transmission, every receiver updates it \emph{want} and \emph{has}
set, based on the packet it has received. $R_0$ continues
transmitting packets until $|H(R_i)|=|P|$, $\forall i$. Therefore
our considered problem of multicasting $M$ packets to $N$ receivers
over an erasure channel is a general case of the index coding
problem. Let $T$ denote the number of transmissions necessary before
all the $N$ receivers receive $M$ innovative packets. The problem
statement which we are interested to solve can be written as an
optimization problem

\begin{equation}\label{eq:optimisation0}
\text{minimise } T
\end{equation}

subject to

\begin{equation}\label{eq:subject}
|H(R_i)| = |P|, \forall i.
\end{equation}

In this paper we propose a triangular packet coding scheme which
solves \eqref{eq:optimisation0} optimally, with encoding and
decoding computational complexity of the same order as that of
traditional XOR packet encoding and decoding.

Therefore the main contribution of our paper is twofold, first we
address the energy cost and bandwidth performance tradeoff
associated with network coding~\cite{Heide11} in this paper.
Secondly, from the information theory perspective, we show that
unlike previous works which have concluded that the optimal solution
for the index coding problem using over $GF(2)$ is
NP-complete~\cite{Kwan11,KChi10,Rouayheb07,Rozner07}, it is possible
to obtain an optimal solution for wireless multicasting using XOR
coding when the constraint of adding redundant bits to packets is
relaxed.

\subsection{Performance Bound}
\newtheorem{theorem}{Theorem}
\newtheorem{lemma}{Lemma}
\newtheorem{axiom}{Axiom}

The main performance measure for this work is the number of
transmissions required to transmit $M$ innovative packets to $N$
memory-based receivers. The performance bound can be derived when an
optimal network coding scheme is considered, where every packet
transmission is an innovative transmission. Similar
to~\cite{KChi10,Nguyen09}, when packet reception is characterized by
the binomial probability law, with nonhomogeneous packet loss
probabilities, the average number of transmissions required to
transmit $M$ innovative packets to $N$ memory-based receivers is
\[
\mathcal{L} = \sum_{n=0}^{\infty} \left\{ 1 - \prod_{j=1}^{N} \left(
\sum_{i=M}^{n} {n \choose i} (1-p_j)^{i} p_j^{n-i} \right) \right\}
\]
where $p_j$ is the packet loss probability for receiver $R_j$.

We may approximate the above result by only including those whose
packet loss are high. With this approximation, the above result can
be rewritten as
\[
\mathcal {G}(p,k,M) = \sum_{n=0}^{\infty} \Biggr\{ 1 - \left(
\sum_{i=M}^{n} {n \choose i} (1-p)^{i} p^{n-i} \right)^k \Biggr\} ,
\]
where $p=\max\{p_i\}$ and $k$ is the number of receivers with packet
loss probability $p_i = p$.

\begin{axiom}\label{axiom0}
For a network with specified $M$, $k$ and $p$ any network coding
scheme which can transmit $\mathcal {G}(p,k,M)$ innovative packets
in $\mathcal {G}(p,k,M)$ transmissions is considered an optimal
coding scheme.
\end{axiom}

\section{Proposed Coding Scheme} \label{sect:proposed}
We first illustrate the practical usefulness of our proposed coding
scheme by the aid of a simple example. Consider the case of $R_0$
multicasting packets $c_1$ and $c_2$ to 4 receivers, $R_1,...,R_4$.
After the first 2 transmissions, only $R_1$ receives $c_1$ and only
$R_2$ receives $c_2$. Now to transmit a packet which is innovative
for all the receivers, the only possibility is to transmit
$c_1\oplus c_2$, when coding is limited to $GF(2)$. However let us
assume that only $R_3$ receives this coded packet. Given the
constraint of coding over $GF(2)$ it is easy to verify that after
these 3 transmissions, there is no possibility to transmit a packet
which will be innovative for all the receivers. This is also
consistent with previous theoretical analysis~\cite{KChi10,Kwan11},
which have proven that an innovative packet transmission, for every
transmission is only possible on the condition that the field size
is larger than or equal to the number of users.

In our proposed coding scheme we go around this information theory
limitation, by adding redundant bits to $c_1$ and $c_2$, then
encoding these packets over $GF(2)$, and including information about
these redundant bits added to each packet in the packet header of
the encoded packet. For the given example, we add bit `0' at the
head of data payload of packet $c_1$. To equalize the length of both
the packets we add bit `0' at the tail of the data payload of packet
$c_2$, and then encode these packet over $GF(2)$. Let us assume that
all packets $c_m\in P$, $1\leq m\leq M$, have equal data payload
length of $B$ bits. The binary bit pattern of the data payload for
packet $c_m$ can be represented as $(b_{1,m},b_{2,m},...,b_{B,m}),
b_{j,m}\in\{0,1\}$. Hence the bit pattern of $c_1$ with one
redundant bit added at the head of the packet is given as
$(0,b_{1,1},...,b_{B,1}),$ and that of $c_2$ as
$(b_{1,2},...,b_{B,2},0)$. We denote such modified packet as
$c_{m,r_m}$, where $r_m\in \mathbb{N}_0$, is the number of `0's
added at the head of the data payload of packet $c_m$. The new
encoded packet is denoted as $c_{1,1}\oplus c_{2,0}$. The packet
header will include information about $r_m$ for each packet used for
encoding, and we will study the overhead cost in subsequent section.

This new encoded packet $c_{1,1}\oplus c_{2,0}$ will be innovative
for all the 4 receivers. If $R_3$ receives this packet, it will have
the information about bit $b_{1,2}$ as bit `0' was added as a
redundant bit in packet $c_1$ from the packet header. Since $R_3$
also has packet $c_1\oplus c_2$, it now has information about bit
$b_{1,2}$ and $b_{1,1}\oplus b_{1,2}$. Using this information it can
decode bit $b_{1,1}$. Bit $b_{1,1}$ is then substituted in
$b_{1,1}\oplus b_{2,2}$, from the encoded packet $c_{1,1}\oplus
c_{2,0}$, to obtain bit $b_{2,2}$. Therefore using this bit-by-bit
simple back substitution method, $R_3$ can decode all the bits of
packet $c_1$ and $c_2$.

\subsection{Generation of Innovative Packets}
For a system such the total number of redundant bits added is
$r_{max}$ to a packet of length $B$ bits, then we note that
generating $c_{m,r_m}$ by adding of redundant bits to $c_m$ is
equivalent to the following operation
\[
c_{m,r_m} = 2^{r_{max}-r_m} c_m.
\]

An encoded packet $c_{1,r_1} \oplus c_{2,r_2} \oplus \cdots \oplus
c_{M,r_M}$ is said to be innovative to a receiver if the encoded
packet is linearly independent with respect to all other encoded
packets that the receiver already possesses. Our aim is find a
sequence of coefficient sets such that an encoded packet with a
coefficient set is always linearly independent to encoded packets of
any collection of other coefficient sets, and secondly, the decoding
procedure can be solved using the back substitution method.

Assuming encoding of $M$ packets, for an encoded packet $y_a =
c_{1,r_1} \oplus c_{2,r_2} \oplus \ldots \oplus c_{M,r_M}$, the
total number of redundant bits added to each packet $c_m$ is
$r_{max} = \max\{r_1,r_2,...,r_M\}$. We represent the \emph{unique
id} of encoded packet $y_a$ as $(r_1,r_2,...,r_M)_a$. The encoded
packet can be written as
\[
2^{r_{max}-r_1} c_1 \oplus 2^{r_{max}-r_2} c_2 \oplus \ldots \oplus
2^{r_{max}-r_m} c_m.
\]

We shall propose a natural number based triangular substitution
based network coding. The Gaussian elimination method consists of 2
major steps, the matrix triangularization step and the back
substitution step. The matrix triangularization step accounts for
$\mathcal{O}(M^3)$ steps, while the back substitution step requires
only $\mathcal{O}(M^2)$ steps~\cite{Fraleigh87}. Therefore our aim
is to design a coding scheme such that the receiver does not have to
perform the triangularization step. In other words, with our
proposed coding scheme, the encoding packets are linearly
independent and back substitution ready.

If packet encoding is performed such that for an encoded packet,
$y_a$ with an id of $(r_1,r_2,...,r_M)_a$, $r_m$ is mapped to an
element of the set of natural number sequence given as
$\{0,1,...,(M-1)\}$ using a bijective function, then a collection of
encoded packets can easily form a triangular augmented matrix.
Without loss of generality, assume that the first encoded packet
(also called the initial pattern) has an id given as
$(0,1,2,3,...,M-1)_1$. Then to generate the next encoded packet, we
anchor `0' at its position and rotate all other terms rightward
which results in $(0,M-1,1,2,...,M-2)_2$. Similarly the consecutive
encoded packets in the series are generated as
$(0,M-2,M-1,1,...,M-3)_3$, $(0,M-3,M-2,M-1,...,M-4)_4$, ...,
$(0,2,3,4,...,M-1,1)_{M-1}$. We call all encoded packets in this
series, with the `0' anchored in one position, to form a
\emph{group}. We may now anchor `0' at the second position starting
with an id $(1,0,2,3,...,M-1)_M$ for the next group. By rotating all
other terms except `0', we get another group of $M-1$ encoded
packets. Given $M$ positions for `0' to anchor, we yield altogether
$M$ groups of $M-1$ encoded packets with the very last encoded
packet being $(2,3,4,...,M-1,1,0)_{M(M-1)}$.

But what if the transmitter needs to generate additional encoded
packets? In such case, after the $M(M-1)$ permutations has been
exhausted, the transmitter may start with another initial pattern of
$(0,2,4,6,...,2M-2)_{M(M-1)+1}$ which is derived from the earlier
initial pattern with each $r_m$ multiplied by a constant $\alpha$.
Let $\alpha=2$, and we call this collection of encoded packets the
second \emph{round}. Similarly, the transmitter can generate another
$M(M-1)$ encoded packets for the third round by setting $\alpha=3$.
Theoretically, we may continue with $\alpha=3,4,5,...$ without a
limit, only that a higher value of $\alpha$ results in a larger
value of $r_{max}$ which implies more redundant bits for the
encoding.

In the following, we shall show that our proposed coding scheme
always generate innovative packets. In other words, the coefficient
matrix of any given $n$ encoded packets where $n\leq M$ gives a rank
of $n$. While our proposed coding scheme mixes standard arithmetic
(i.e. multiplication by shifting) and finite field arithmetic (i.e.
exclusive-or), use of standard arithmetic is sufficient to prove
that every generated packet is innovative.

\begin{lemma} \label{lemmma:samegroup}
Consider a system of $M$ packets. Given an encoded packet $y_a$ with
an id of $(0,1,2,3,...,M-1)_a$, the coefficient matrix formed by all
$M-1$ encoded packets in the same group of $y_a$ including $y_a$
gives a rank of $M-1$.
\end{lemma}

\begin{proof}
We first list all the ids of $M-1$ encoded packets in the same group
of $y_a$ including $y_a$ as follows where
\[
\begin{array}{ccccl}
(0,&1,  &2,  &3,...,&M-1)_a \\
(0,&M-1,&1,  &2,...,&M-2)_{a+1} \\
(0,&M-2,&M-1,&1,...,&M-3)_{a+2} \\
   &    &\vdots           \\
(0,&2,  &3,  &4,...,&1)_{a+M-2}. \\
\end{array}
\]

From the id of $y_a$, we know that $r_{max}=M-1$ and $\alpha=1$. The
above list of encoded packets can be expressed by
\[
\begin{array}{rrrrrrr}
2^{M-1} c_1 &+& 2^{M-2} c_2 &+& 2^{M-3} c_3 &+ \cdots &+ 2^{0} c_M\\
2^{M-1} c_1 &+&   2^{0} c_2 &+& 2^{M-2} c_3 &+ \cdots &+ 2^{1} c_M\\
2^{M-1} c_1 &+&   2^{1} c_2 &+& 2^{0}   c_3 &+ \cdots &+ 2^{2} c_M\\
\vdots& \ddots &\vdots \\
2^{M-1} c_1 &+& 2^{M-3} c_2 &+& 2^{M-4} c_3 &+ \cdots &+ 2^{M-2} c_M.\\
\end{array}
\]

For each pair of consecutive equations above, we multiply the upper
one by $2^\alpha$ and then minus with the lower one. After this
simplification, we get
\[
\begin{array}{rrrr}
2^{M-1} c_1 &+& (2^{M-1}-1) c_2 \\
2^{M-1} c_1 &+& (2^{M-1}-1) c_3 \\
2^{M-1} c_1 &+& (2^{M-1}-1) c_4 \\
            &\vdots \\
2^{M-1} c_1 &+& (2^{M-1}-1) c_M \\
\end{array}
\]
which clearly gives a rank of $M-1$.

\vspace{-3mm}
\end{proof}

For other groups, we may reorder the coefficients such that $r_1=0$
and then repeat Lemma~\ref{lemmma:samegroup} to show that other
groups possess the same property. Furthermore,
Lemma~\ref{lemmma:samegroup} can be easily extended to the case of
$\alpha=2,3,...$ which applies to a group belonging to higher
rounds.

\begin{lemma} \label{lemmma:diffgroup}
Consider a system of $M$ packets. A collection of all $M-1$ encoded
packets in the same group and an encoded packet from a different
group gives a rank of $M$.
\end{lemma}

\begin{proof}
We shall prove by calculating the determinant of the matrix. Without
loss of generality, we consider a group, $\mathbf{G}$, of encoded
packets starting with the initial pattern of $(0,1,2,3,...,M-1)_1$.
Based on Lemma~\ref{lemmma:samegroup}, after simplification, these
encoded packets can be represented by
\[
\begin{array}{rrrr}
2^{M-1} c_1 &+& (2^{M-1}-1) c_M \\
2^{M-1} c_1 &+& (2^{M-1}-1) c_{M-1} \\
2^{M-1} c_1 &+& (2^{M-1}-1) c_{M-2} \\
            &\vdots \\
2^{M-1} c_1 &+& (2^{M-1}-1) c_2 \\
\end{array}
\]

Let us consider an arbitrary encoded packet, $y_a$, with the
following representation
\[
y_a = \lambda_1 c_1 + \lambda_2 c_2 + \lambda_3 c_3 + \cdots +
\lambda_M c_M.
\]
Combining the last two results, we form a coefficient matrix and
express its determinant, $d$ as
\[
\left|
\begin{array}{ccccccc}
2^{M-1} & 0 & 0 & ... & 2^{M-1}-1 \\
&&\vdots \\
2^{M-1} & 0 & 2^{M-1}-1 & ... & 0 \\
2^{M-1} & 2^{M-1}-1 & 0 & ... & 0 \\
\lambda_1 & \lambda_2 & \lambda_3 & ... & \lambda_M \\
\end{array}
\right|.
\]

The condition for $y_a$ to be linearly independent of the encoded packets from
group $\mathbf{G}$ is that $d$ produces a none zero value. After some algebraic
manipulation on $d$ and applying the condition where $d\neq0$, we yield the
following condition where
\begin{equation}\label{eqn:cond}
2^{M-1} \sum_{i=2}^M \lambda_i \neq (2^{M-1}-1) \lambda_1.
\end{equation}

Our coding scheme demands the one-to-one mapping between
$\{r_1,r_2,...,r_M\}$ and $\{0,1,...,(M-1)\}$. In other words,
$\{\lambda_1,\lambda_2,...,\lambda_M\}$ are mapped one-to-one to
$\{2^{(M-1)},2^{(M-2)},...,2^0\}$. The only case that
\eqref{eqn:cond} cannot be met occurs when $\lambda_1=2^{(M-1)}$. In
other cases, we always have
\[
\sum_{i=2}^M \lambda_i > \lambda_1 > 0
\]
and since $2^{M-1}>(2^{M-1}-1)>0$, condition \eqref{eqn:cond} is
always met.

The setting of $\lambda_1 \neq 2^{(M-1)}$ requires that $r_1 \neq 0$ for $y_a$
of id $(r_1,r_2,...,r_M)$. As $y_a$ is drawn from a group other than group
$\mathbf{G}$, its `0' is anchored at a different position rather than at $r_1$
like those in group $\mathbf{G}$, thus we have $r_1 \neq 0$, and hence $y_a$ is
linearly independent of the encoded packets from group $\mathbf{G}$.

\vspace{-3mm}
\end{proof}

Applying a quick test using \eqref{eqn:cond} in
Lemma~\ref{lemmma:diffgroup} shows that based on our coding scheme,
a pick of an encoded packet even from a group of other rounds
satisfies the condition. Consequently, with the above lemmas, we
established that every encoded packet is an innovative packet.

\subsection{Packet Decoding}
\begin{figure}
\begin{center}
\includegraphics[width = 0.5\textwidth]{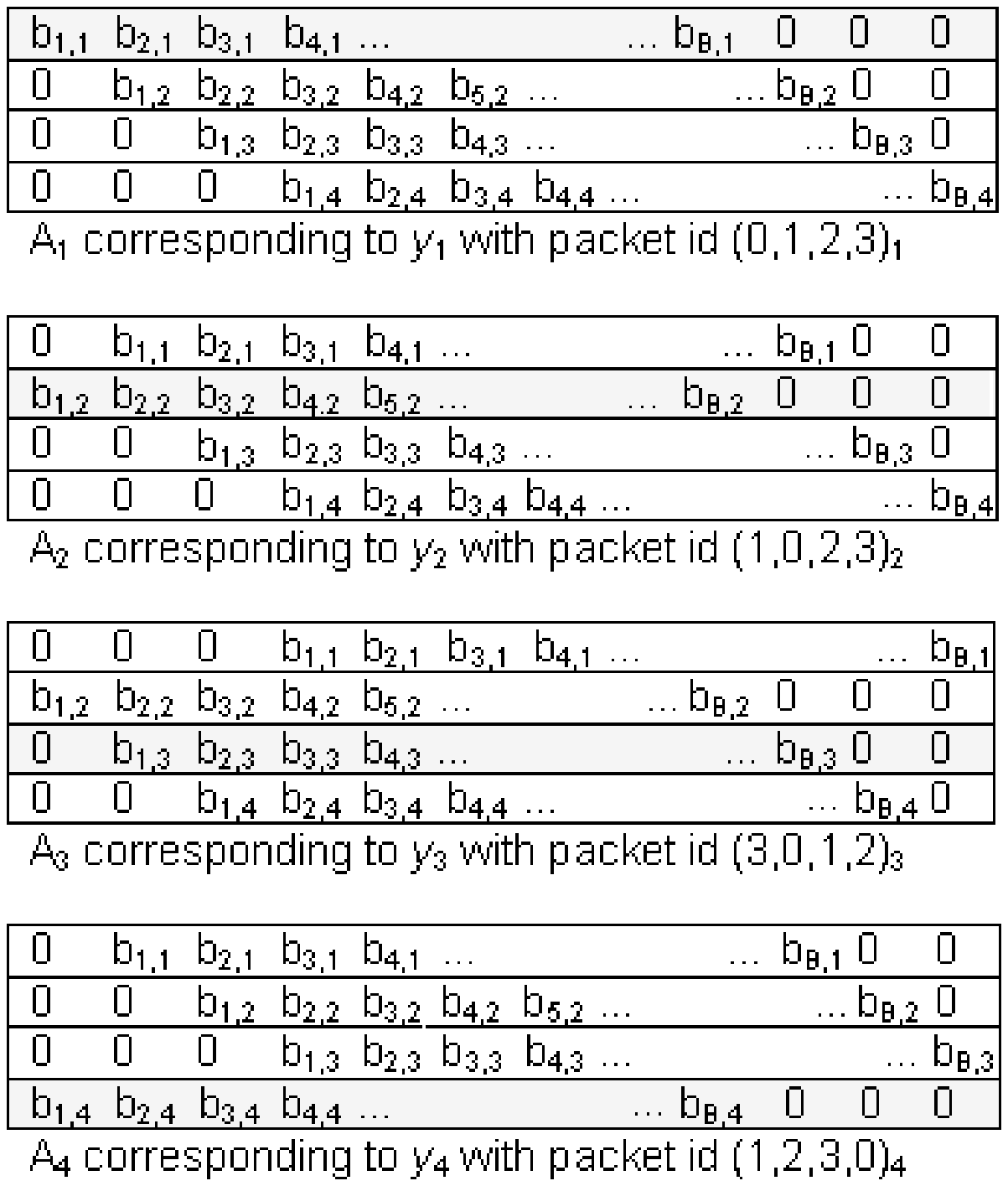}
\end{center}
\caption{An illustrating example to show how packet decoding is done
on triangular network coding scheme. Highlighted $b_{j,m}$ represent
calculated bit value from an equation with only 1 unknown variable,
whereas non-highlighted $b_{j,m}$ represent substituted bit value.
Bits `0' are known bit values from the packet's header. Packet $y_2$
and $y_3$ are from the same group.} \label{fig:packets}
\end{figure}

We now show that such a coding scheme can be solved by substitution
method. We illustrate the concept using an example depicted in
Fig.~\ref{fig:packets}. In our example, we consider a set of 4
original packets, that is $M=4$. After transmissions of several
encoded packets, we consider that a receiver has successfully
collected 4 encoded packets, $y_1$, $y_2$, $y_3$ and $y_4$ with the
packet ids of $(0,1,2,3)_1$, $(1,0,2,3)_2$, $(3,0,1,2)_3$ and
$(1,2,3,0)_4$ respectively. In Figure~\ref{fig:packets}, each table
represents an encoded packet where each row lists the bits of an
original packet involved in the encoding. It can be seen that the
first bit of $y_1$ counting from the left is encoded by $b_{1,1}
\oplus 0 \oplus 0 \oplus 0$ which equals $b_{1,1}$. Similarly,
$b_{1,2}$ and $b_{1,4}$ can be obtained from the first bits of $y_2$
and $y_4$ respectively. Now, the decoder may proceed to the second
bit position of the 4 packets. Substituting $b_{1,1}$, $b_{1,2}$ and
$b_{1,4}$ into the $(M-1)$ matrices, $b_{2,1}$, $b_{2,2}$, $b_{2,4}$
and $b_{1,3}$ can be obtained immediately. Moving on to the bit
position, bits $b_{3,1}$, $b_{3,2}$, $b_{2,3}$ and $b_{3,4}$ can be
immediately solved by substitution. The process can continue further
until all unknown bits are solved. By this way, a receiver can
decode the bits of all 4 packets through the simple back
substitution at the bit level.

\section{Performance Evaluation}\label{sec:performance}

\subsection{Throughput Performance}
Based on Axiom~\ref{axiom0}, if for a network with specified $M$,
$k$ and $p$, our proposed coding scheme can generate at least
$\mathcal {G}(p,k,M)$ innovative packets, then it can be considered
as an optimal coding scheme. It is easy to verify that for

\begin{equation}\label{eq:inequality0}
\alpha (M^2-M)  \geq \mathcal {G}(p,k,M)
\end{equation}
there always $\exists\alpha$ such that the above inequality is
satisfied. As we will show later in Section~\ref{sect:results}, for
most practical setting (generally corresponding to $M\geq 10$)
$\alpha=1$ is sufficient to guarantee at least $\mathcal {G}(p,k,M)$
innovative packets. For smaller $M$, a small value of $\alpha$ is
sufficient, and therefore our proposed coding scheme does not have
significantly high packet overhead. In practice, a very small value
of $M$ is often not desirable as it does not provide sufficient pool
of packets to utilize the benefit of network coding. We formally
study the packet overhead cost of our proposed coding scheme in
subsequent subsection.

The throughput performance of DLNC is optimal when $q\geq \lceil
\log_2 N\rceil$. However in the paper~\cite{KChi10}, the optimal
bound and DGC (Dynamic General Coding)\footnote{DGC is the name of a
DLNC based algorithm demonstrated in~\cite{KChi10}.} throughput
performance results do not match. This is explained as follow. The
throughput bound the authors consider is not tight, however since
the DGC algorithm claims to be able to always find coding
coefficients which is linearly independent for all the receivers,
implies that the DGC is an optimal coding scheme.

Assuming fixed $N$, sparse LNC can be solved in polynomial time to
generate sparse and linearly independent coding coefficients,
otherwise optimal throughput for sparse LNC is NP-complete. However
since solving SPARSITY for a fixed network size is a special case of
the SPARSITY problem, therefore in the general sense the throughput
optimality of sparse LNC is considered NP-complete.

\subsection{Packet Overhead}

For LNC, the packet header needs to include information about
encoding coefficient $g_m$ used to multiply with $c_m$. If $g_m\in
GF(2^q)$, the number of bits required to represent $g_m$ is given as
$\log_2 2^q$. Since there are $M$ such encoding coefficients, the
total packet overhead of RLNC is given as $M \log_2 2^q$, which
reduces to $Mq$ bits. For an optimal DLNC and sparse coding, $q$ is
bounded by $q\geq \lceil \log_2 N\rceil$~\cite{KChi10,Sung11},
therefore the total packet overhead for an optimal DLNC and sparse
coding scheme is given as $M\lceil\log_2 N\rceil$.

For our proposed coding scheme, the packet overhead is given by
$r_{max}$ redundant bits added to every packet, and $M \lceil \log_2
r_{max}\rceil$ bits to store the encoded packet unique id
$(r_1,...,r_M)_a$ in the encoded packet header. Therefore the total
packet overhead of our proposed coding scheme is given as $r_{max}+M
\lceil \log_2 r_{max}\rceil$ bits, where $r_{max}=\alpha(M-1)$, and
$\alpha$ corresponds to the smallest value required to satisfy
inequality \eqref{eq:inequality0}. For $\alpha=1$ (which satisfies
inequality~\eqref{eq:inequality0} when $M$ is reasonably large), the
total packet overhead is $M-1+M\lceil\log_2(M-1)\rceil$ or simply
$M+M\lceil\log_2M\rceil$ if we ignore the insignificant constant
term of minus one.

We would like to further point out that the packet overhead of
triangular network coding can be reduced to approximately $M$ bits,
which is the number of redundant bits added to each packet. The only
information which needs to be included in the packet header is the
packet batch size, the group index and packet index. Since
triangular network coded packets' id follow a sequence, with these
information in the packet header, the receiver can reconstruct the
packet id of the coded packet.

However for our coding comparison we assume the packet overhead of
$M+M\lceil\log_2M\rceil$ for triangular network coding to provide
fair comparison with RLNC packet overhead, since assuming that both
the transmitter and receiver have the same random number generator
(RNG), all what a RLNC coded packet then needs to include in the
packet header is the seed and packet index. Using the seed and
packet index, the receiver can then regenerate the random coding
coefficients.

\subsection{Computational Complexity}
We evaluate three distinct computational complexity. The
\emph{algorithm complexity} refers to the complexity of generating
the coding coefficients for $M$ coded packets. The total algorithm
complexity is given as $M$ times the complexity of finding $M$
coding coefficient for each coded packet. The \emph{encoding
complexity} refers to the complexity of encoding $M$ coded packets.
Whereas the \emph{decoding complexity} signifies the complexity of
decoding $M$ innovative coded packets by a receiver.

We evaluate the computational complexity for $M$ coded packets for
uniformity in comparison, though in practise the total algorithm and
encoding complexity would also be dependent on the throughput
optimality of the coding scheme. Optimal coding scheme will need to
run the algorithm and perform coding $\mathcal {G}(p,k,M)$ times
only, whereas suboptimal coding schemes will need to run the
algorithm and perform coding $\mathcal {G}(p,k,M)+\mathcal {C}$
times, where $\mathcal {C}$ is a coding scheme dependent arbitrary
constant.

The decoding complexity is more significant than the algorithm and
encoding complexity. This is because the algorithm and encoding
computation cost is borne by a single transmitter which is often not
a battery constrained device such as an AP, whereas the decoding
cost is borne by $N$ receivers, which are often battery and
processor constrained such as smartphones.

We evaluate the encoding and decoding complexity at bit level, so
that this way we can distinguish between the complexity of
multiplication used in LNC and addition used in XOR coding. For two
non-negative integers $a$ and $b$, the multiplication and addition
complexity is given as $\mathcal {O}((\log_2 a) (\log_2 b))$ and
$\mathcal {O}(\log_2 a+\log_2 b)$ respectively~\cite{Menezes96}.
While there exist different optimization algorithms to reduce the
computational complexity of multiplication, such optimization survey
is beyond the scope of the paper.

\subsubsection{Traditional XOR}
The encoding complexity of a traditional XOR packet encoding of $m$
packets\footnote{XOR coding is often performed on a subset of
packets from $P$~\cite{Katti06,Rozner07}, where $|P|=M$, the
cardinality of the subset of packets over which encoding is
performed varies for each encoded packet and is algorithm
dependent.} is given as $\mathcal {O}(mB)$, where $m\leq M$. XOR
encoding is performed bit-by-bit. For each bit ($b_{j,m}$) position
it needs to evaluate the value of $m$ bits in location $j$.
Therefore the average computational complexity to generate $M$ coded
packets is given as $\mathcal {O}(MmB)$. Similarly XOR-encoded
packet can be decoded by performing the XOR operation on a set of
packets, which has complexity given as $\mathcal {O}(mB)$. To decode
$M$ encoded packets, the average computational complexity is given
as $\mathcal {O}(MmB)$.

For the algorithm complexity we consider the sort-by-utility
algorithm~\cite{Rozner07} for our evaluation, which has complexity
given as $\mathcal {O}(M^2\log_2 M)$. We have chosen the
sort-by-utility algorithm, as it has been shown to be an efficient
XOR coding algorithm with respect to other known XOR coding
algorithms in~\cite{Rozner07}.

\subsubsection{RLNC}
RLNC needs to perform random number generation, multiplication and
addition during encoding. Generating a random number has constant
time complexity of $G$. Therefore the total algorithm complexity is
given as $\mathcal {O}(M^2G)$. The operation $g_m\cdot c_m$ has
complexity given as $\mathcal {O} (qB)$. To generate a coded packet
$M$ such operations needs to be performed, hence the complexity to
generate a coded packet is given as $\mathcal {O}(MqB)$. Therefore
the total encoding complexity to generate $M$ coded packets is given
as $\mathcal {O}(M^2qB)$.

For decoding, RLNC needs to perform Gaussian elimination with
complexity given as $\mathcal {O}(M^3)$, and then multiplication and
addition of $M$ coded packet with the inverse matrix. Therefore the
total decoding complexity is given as $\mathcal {O}(M^2qB+M^3)$.

\subsubsection{DLNC}
The only known DLNC algorithm is the DGC algorithm given
in~\cite{KChi10}. DGC has algorithm complexity given as $\mathcal
{O}(N^2M^3)$ for one coded packet. Its encoding and decoding
complexity is the same as that of RLNC.

\subsubsection{Sparse LNC}
Sparse LNC is a special case of DLNC. Unlike DLNC, where the main
objective is to find coding coefficients which are linearly
independent from received coding coefficients by all the receivers,
in sparse LNC, the objective is not only to find linearly
independent coding coefficients, but also coefficients which are
sparse. The algorithm complexity for sparse LNC is given as
$\mathcal {O}(M^N(MN^2))$, which reduces to $\mathcal {O}(M^NN^2)$.
Since there are large number of `0' coding coefficients in sparse
LNC we can ignore the computational cost of $g_m\cdot c_m$ when
$g_m=0$. Hence the coding coefficient is given as $\mathcal
{O}(Mq\omega B)$, where $\omega$ is the number of non-zero
components, $\omega\leq M$. The decoding complexity is given as
$\mathcal {O}(M^2qB+\omega M^2)$.

\subsubsection{Triangular Network Coding}
The number of flops (or number of steps) for our proposed triangular
network coding scheme is given as follow. Generation of unique
packet ids for packet encoding has constant time complexity, as
these packet id follow a natural number based sequence. The
algorithm to generate the coefficient matrix for $M$ packets has
complexity of $\mathcal {O}(M^2)$. For encoding complexity, the
encoding of each packet has complexity of $\mathcal {O}(MB)$, since
there are $B$ bits from $M$ packets, and XOR addition is the only
required operation. The redundant `0' bits can be ignored both
during encoding and decoding. Therefore to generate $M$ such encoded
packet, the overall complexity is given by $\mathcal {O}(M^2B)$.

Once a receiver has $M$ innovative encoded packets, then the
complexity to decode these $M$ packets is given as $\mathcal
{O}(M^2B)$. We elaborate on the back-substitution process. $M$ flops
are required to find the solution of each bit (XOR addition), the
solution for this equation is then substituted in $M-1$ matrices.
Since there are $M$ such equations the total number of flops is
given as $M(2M-1)$. Hence the decoding computational complexity for
each bit location is given as $\mathcal {O}(M^2)$, which is also
consistent with the results shown in~\cite{Fraleigh87}.

Our proposed coding scheme results in triangular pattern both at the
start and end of the packet (See Figure~\ref{fig:packets} for
illustration), and therefore for a multiprocessor system, it is
possible that one processor decodes the first half part of the
packets from the start, while simultaneously another processor
decodes the second half part from the end of the packet (reverse
direction) using the substitution method. This can cut down the
decoding time by half in a multiprocessor system.

\subsection{Packet Feedback Requirement}
We further note that traditional XOR, DLNC and sparse LNC require
feedback information from all receivers to decide what coefficient
sets to be used for further encoding of packets such that the
encoded packets can be innovative to all receivers. In this aspect,
RLNC and our proposed scheme do not need the feedback information.
For RLNC, it simply generates encoded packets using random numbers.
For our proposed scheme, the transmitter can simply continue the
sequence of coefficient sets as the generated encoded packets will
always be innovative.

Collecting packet feedback from $N$ receivers over an erasure
channel induces a very large overhead. For wireless networks, the
IEEE 802.11-2007 standard does not provide any MAC level reliability
for multicast transmissions. Therefore deterministic network coding
schemes are often based on an impractical arbitrary assumption that
the transmitter has perfect feedback information for all transmitted
packets.

\section{Numerical Results}\label{sect:results}
\begin{figure}
\begin{center}
\includegraphics[width = 0.5\textwidth]{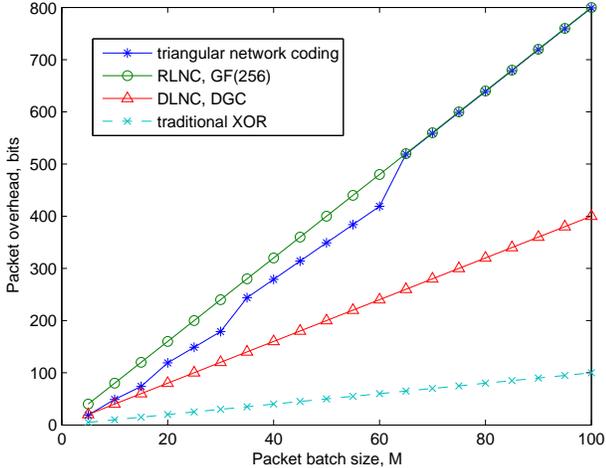}
\end{center}
\caption{Packet overhead cost for N=10 and p=0.3.}
\label{fig:result1}
\end{figure}

\begin{figure}
\begin{center}
\includegraphics[width = 0.5\textwidth]{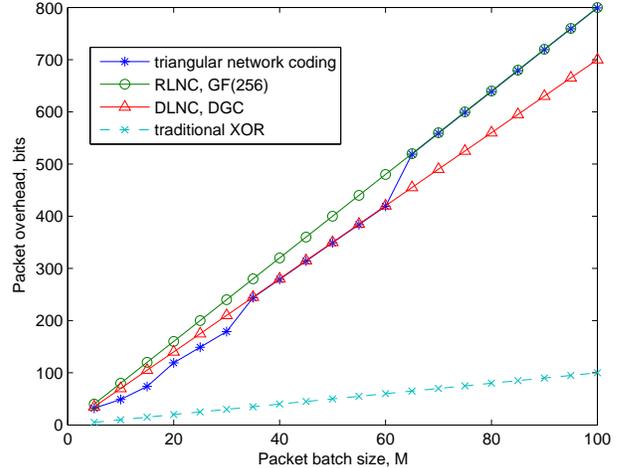}
\end{center}
\caption{Packet overhead cost for N=100 and p=0.8.}
\label{fig:result3}
\end{figure}

We now study the packet overhead corresponding to the smallest
$\alpha$ necessary to satisfy inequality \eqref{eq:inequality0}, and
compare the results with the packet overhead of RLNC. Without loss
of generality, we assume that $p_1=...=p_N=p$, such that $k=N$.
Unlike the packet overhead of triangular coding, which is a
monotonic function with respect to $M$, RLNC and DLNC packet
overhead cost are a function of two variables.

For RLNC there is no standardized field size suggested, though an
overwhelming majority of works have considered a field size of
$GF(2^8)$ for RLNC as being sufficiently large enough to guarantee
linear independency with very high
probability~\cite{Shojania09,Heide11,Vingelmann10,Sagduyu07}. We
therefore assume $q=8$ for the evaluation of RLNC packet overhead.

The result of packet overhead for various coding scheme for
different $M$, $N$, and $p$ is given in Figs.~\ref{fig:result1} and
\ref{fig:result3}. The results show that for practical setting, the
packet overhead cost of our proposed coding scheme is lower than
that of RLNC. For a large network (see Fig.~\ref{fig:result3}), the
packet overhead cost of DLNC is similar as that of triangular
coding. Even for a large network with a high packet loss probability
as shown in Fig.~\ref{fig:result3}, the packet overhead cost of
triangular network coding is the same as that of a smaller network
with low packet loss probability. In fact the only difference
between the results presented in Fig.~\ref{fig:result1} and
Fig.~\ref{fig:result3} occurs at $M=5$ for triangular network
coding, when $\alpha=2$ is required to satisfy
inequality~\ref{eq:inequality0}, which corresponds to a slightly
higher packet overhead cost.

The sudden increase of packet overhead cost for triangular
substitution from $M=30$ to $M=35$ and $M=60$ to $M=65$ is explained
from the presence of a ceiling logarithm to base two term in the
packet overhead cost of our proposed scheme.

\section{Conclusion}\label{sect:conclusion}

\begin{center}
\begin{table*}[ht]
\caption{Summary of the characteristics of various coding schemes.}
\label{table:summary} {\hfill{}
\begin{tabular}{|c|c|c|c|c|c|}
\hline & RLNC  &  DLNC (DGC)~\cite{KChi10} & Sparse LNC~\cite{Sung11}& Traditional XOR  & Triangular Network Coding\\
\hline Throughput & Suboptimal & Optimal & NP-complete & NP-complete & Optimal\\
\hline Algorithm Complexity & $\mathcal {O}(M^2G)$ & $\mathcal
{O}(N^2M^4)$ & $\mathcal {O}(M^NN^2)$ & $\mathcal
{O}(M^2\log_2 M)$~\cite{Rozner07} & $\mathcal {O}(M^2)$\\
\hline Encoding Complexity & $\mathcal {O}(M^2qB)$ & $\mathcal
{O}(M^2qB)$ & $\mathcal {O}(Mq\omega B)$ & $\mathcal {O}(MmB)$
& $\mathcal {O}(M^2B)$\\
\hline Decoding Complexity & $\mathcal {O}(M^2qB+M^3)$ & $\mathcal
{O}(M^2qB+M^3)$ & $\mathcal {O}(M^2qB+\omega M^2)$ & $\mathcal
{O}(MmB)$ & $\mathcal {O}(M^2B)$\\
\hline Packet Overhead (bits) & $Mq$ & $M\lceil\log_2N\rceil$ &
$M\lceil\log_2N\rceil$ & $M$ & $M+M\lceil\log_2M\rceil$\\
\hline Packet Feedback & Not required & Required & Required &
Required & Not required\\
\hline
\end{tabular}}
\hfill{}
\end{table*}
\end{center}

The characteristics of various coding schemes, including triangular
coding are summarized in Table~\ref{table:summary}. In this work we
have proposed an optimal triangular network coding scheme over
$GF(2)$, and demonstrated its throughput and computation cost merits
over previously proposed network coding schemes. Unlike traditional
linear network coding scheme which requires the use of
multiplication and Gaussian elimination method at the encoder and
decoder respectively, our proposed coding scheme uses XOR addition
and back-substitution method at the encoder and decoder
respectively. This way, rather than decreasing the number of
computation steps, we decrease the \emph{order} of computational
complexity for both encoding and decoding, without compromising on
the throughput performance. Such a coding scheme is throughput
optimal, and its packet overhead increases linearithmic with respect
to $M$, which can be optimized such that the packet overhead
increases linearly with respect to $M$. Most importantly, triangular
network coding performance is independent of the packet feedback
information.

In addition, apart from the practical implementation significance of
our coding scheme, such results are also of interest from an
information theoretic studies perspective, as unlike previous works,
we have shown that it is possible to obtain optimal coding solution
over $GF(2)$ by adding few redundant bits in the packets.

Such a coding scheme would in particular be of interest in energy
and processor constraint devices. For our future work, we would like
to study the feasibility of extending our current coding scheme for
a distributed multi-hop wireless network.

\bibliographystyle{IEEEtran}
\bibliography{IEEEabrv,mainJ}

\begin{thebibliography}{10}
\providecommand{\url}[1]{#1}
\csname url@rmstyle\endcsname
\providecommand{\newblock}{\relax}
\providecommand{\bibinfo}[2]{#2}
\providecommand\BIBentrySTDinterwordspacing{\spaceskip=0pt\relax}
\providecommand\BIBentryALTinterwordstretchfactor{4}
\providecommand\BIBentryALTinterwordspacing{\spaceskip=\fontdimen2\font plus
\BIBentryALTinterwordstretchfactor\fontdimen3\font minus
  \fontdimen4\font\relax}
\providecommand\BIBforeignlanguage[2]{{%
\expandafter\ifx\csname l@#1\endcsname\relax
\typeout{** WARNING: IEEEtran.bst: No hyphenation pattern has been}%
\typeout{** loaded for the language `#1'. Using the pattern for}%
\typeout{** the default language instead.}%
\else
\language=\csname l@#1\endcsname
\fi
#2}}

\bibitem{Yossef06}
Z.~Bar-Yossef, Y.~Birk, T.~S. Jayram, and T.~Kol, ``{Index Coding with Side
  Information },'' in \emph{IEEE FOCS}, Berkeley, USA, October 2006.

\bibitem{Chaudhry08}
M.~Chaudhry and A.~Sprintson, ``{Efficient Algorithms for Index Coding},'' in
  \emph{IEEE INFOCOM Workshop}, Phoeniz, USA, April 2008.

\bibitem{Katti06}
S.~Katti, H.~Rahul, W.~Hu, D.~Katabi, M.~Medard, and J.~Crowcroft, ``{XORs in
  the Air: Practical Wireless Network Coding},'' in \emph{ACM SIGCOMM}, Pisa,
  Italy, September 2006.

\bibitem{Sagduyu07}
Y.~E. Sagduyu and A.~Ephremides, ``{On Network Coding for Stable Multicast
  Communication},'' in \emph{IEEE MILCOM}, Florida, USA, October 2007.

\bibitem{KChi10}
K.~Chi, X.~Jiang, and S.~Horiguchi, ``{Network coding-based reliable multicast
  in wireless networks},'' \emph{Elsevier Computer Networks}, vol.~54, no.~11,
  pp. 1823--1836, August 2010.

\bibitem{Heide08}
J.~Heide, M.~V. Pedersen, F.~H.~P. Fitzek, and T.~Larsen, ``{Cautious View on
  Network Coding - From Theory to Practice},'' \emph{Journal of Communications
  and Networks}, vol.~10, no.~4, pp. 403--411, December 2008.

\bibitem{Kwan11}
H.~Y. Kwan, K.~W. Shum, and C.~W. Sung, ``{Generation of Innovative and Sparse
  Encoding Vectors for Broadcast Systems with Feedback},'' in \emph{IEEE ISIT},
  Saint Petersburg, Russia, July 2011.

\bibitem{Rouayheb07}
S.~E. Rouayheb, M.~Chaudhry, and A.~Sprintson, ``{On the Minimum Number of
  Transmissions in Single-Hop Wireless Coding Networks},'' in \emph{IEEE ITW},
  California, USA, September 2007.

\bibitem{Heide11}
J.~Heide, M.~V. Pedersen, F.~H.~P. Fitzek, and M.~Medard, ``{On the Code
  Parameters and Coding Vector Representation for Practical RLNC},'' in
  \emph{IEEE ICC}, Kyoto, Japan, June 2011.

\bibitem{Vallina10}
N.~Vallina-Rodriguez, P.~Hui, J.~Crowcroft, and A.~Rice, ``{Exhausting Battery
  Statistics},'' in \emph{ACM MobiHeld}, New Delhi, India, August 2010.

\bibitem{Fehske11}
A.~Fehske, G.~Fettweis, J.~Malmodin, and G.~Biczok, ``{The Global Footprint of
  Mobile Communications: The Ecological and Economic Perspective},'' \emph{IEEE
  Communications Magazine}, vol.~49, no.~8, pp. 55--62, August 2011.

\bibitem{Shojania09}
H.~Shojania and B.~Li, ``{Random Network Coding on the iPhone: Fact or
  Fiction?}'' in \emph{ACM NOSSDAV}, New York, USA, June 2009.

\bibitem{Vingelmann09}
P.~Vingelmann, P.~Zanaty, F.~H.~P. Fitzek, and H.~Charaf, ``{Implementation of
  Random Linear Network Coding on OpenGL-enabled Graphics Cards},'' in
  \emph{IEEE EW}, Aalborg, Denmark, May 2009.

\bibitem{Vingelmann10}
P.~Vingelmann, M.~V. Pedersen, F.~H.~P. Fitzek, and J.~Heide, ``{Multimedia
  Distribution using Network Coding on the iPhone Platform},'' in \emph{ACM
  MCMC}, Firenze, Italy, October 2010.

\bibitem{Rozner07}
E.~Rozner, A.~Padmanabha, Y.~Mehta, L.~Qiu, and M.~Jafry, ``{ER: Efficient
  Retransmission Scheme For Wireless LANs},'' in \emph{ACM CoNEXT}, New York,
  USA, December 2007.

\bibitem{Shojania07}
H.~Shojania and B.~Li, ``{Parallelized Progressive Network Coding With Hardware
  Acceleration},'' in \emph{IEEE IWQoS}, Evaston, USA, June 2007.

\bibitem{Park10}
K.~Park, J.-S. Park, and W.~W. Ro, ``{On Improving Parallelized Network Coding
  with Dynamic Partitioning},'' \emph{IEEE Transactions on Parallel and
  Distributed Systems}, vol.~21, no.~11, pp. 1547--1560, November 2010.

\bibitem{Darmohray87}
G.~A. Darmohray and E.~D. Brooks, ``{Gaussian Techniques on Shared Memory
  Multiprocessor Computers},'' in \emph{SIAM PPSC}, Los Angeles, USA, December
  1987.

\bibitem{Sung11}
C.~W. Sung, K.~W. Shum, and H.~Y. Kwan, ``{On the Sparsity of a Linear Network
  Code for Broadcast Systems with Feedback },'' in \emph{IEEE NetCod}, Beijing,
  China, July 2011.

\bibitem{Li11}
Y.~Li, E.~Soljanin, and P.~Spasojevic, ``{Effects of the Generation Size and
  Overlap on Throughput and Complexity in Randomized Linear Network Coding },''
  \emph{IEEE Transactions on Information Theory}, vol.~57, no.~2, pp.
  1111--1123, February 2011.

\bibitem{Heide09}
J.~Heide, M.~V. Pedersen, F.~H.~P. Fitzek, and T.~Larsen, ``{Network Coding for
  Mobile Devices - Systematic Binary Random Rateless Codes},'' in \emph{IEEE
  ICC Workshop}, Dresden, Germany, June 2009.

\bibitem{Nguyen09}
D.~Nguyen, T.~Tran, T.~Nguyen, and B.~Bose, ``{Wireless Broadcast Using Network
  Coding},'' \emph{IEEE Transactions on Vehicular Technology}, vol.~58, no.~2,
  pp. 914--925, February 2009.

\bibitem{Fraleigh87}
J.~B. Fraleigh and R.~A. Beauregard, \emph{Linear Algebra}.\hskip 1em plus
  0.5em minus 0.4em\relax Addison-Wesley Publishing Company, 1987, p 91-97.

\bibitem{Menezes96}
A.~Menezes, P.~Oorschot, and S.~Vanstone, \emph{Handbook of Applied
  Cryptography}.\hskip 1em plus 0.5em minus 0.4em\relax CRC Press, 1996, p 66.

\end{thebibliography}

\end{document}